\documentclass[twocolumn,
showpacs,preprintnumbers,amsmath,amssymb]{revtex4}

\usepackage{amssymb,amsmath,amsthm}
\usepackage{graphicx}
\newtheorem{Thm}{Theorem}
\newtheorem{Cor}{Corollary}

\newtheorem{Prop}{Proposition}
\theoremstyle{definition}

\newcommand{\bra}[1]{{\left\langle #1 \right|}}
\newcommand{\ket}[1]{{\left| #1 \right\rangle}}

\newcommand{\T}{\mbox{$\mathrm{tr}$}}

\begin{document}
\title{Generalized entanglement constraints in multi-qubit systems in terms of Tsallis entropy}
\author{Jeong San Kim}
\email{freddie1@khu.ac.kr} \affiliation{
 Department of Applied Mathematics and Institute of Natural Sciences, Kyung Hee University, Yongin-si, Gyeonggi-do 446-701, Korea
}
\date{\today}

\begin{abstract}
We provide generalized entanglement constraints in multi-qubit systems in terms of Tsallis entropy.
Using quantum Tsallis entropy of order $q$, we first provide a generalized monogamy inequality of multi-qubit entanglement
for $q=2$ or $3$. This generalization encapsulates multi-qubit CKW-type inequality as a special case. We further provide a generalized polygamy
inequality of multi-qubit entanglement in terms of Tsallis-$q$ entropy for $1 \leq q \leq2$ or $3 \leq q \leq 4$, which also
contains the multi-qubit polygamy inequality as a special case.
\end{abstract}

\maketitle

\section{Introduction}

Quantum Tsallis entropy is a one-parameter generalization of
von Neumann entropy with respect to a nonnegative real parameter $q$~\cite{tsallis, lv}.
Tsallis entropy is used in many areas of quantum
information theory; Tsallis entropy can be used to characterize
classical statistical correlations inherented in quantum states~\cite{rr},
and it provides some conditions for separability of quantum
states~\cite{ar,tlb,rc}. There are also discussions about using the non-extensive statistical
mechanics to describe quantum entanglement in terms of Tsallis entropy~\cite{bpcp}.

As a function defined on the set of density matrices, Tsallis entropy is concave for all $q > 0$, which
plays an important role in quantum entanglement theory.
Because the concavity of Tsallis entropy assures the property of {\em entanglement monotone}~\cite{vidal},
it can be used to construct a faithful entanglement measure, which does not increase under
{\em local quantum operations and classical communication} (LOCC).

One distinct property of quantum entanglement from other classical correlations
is that multi-party entanglement cannot be freely shared among the parties.
This restricted shareability of entanglement in multi-party quantum systems is known as
{\em monogamy of entanglement}(MoE)~\cite{T04, KGS}.
MoE is a key ingredient for secure quantum cryptography~\cite{rg,m},
and it also plays an important role in condensed-matter physics such as the $N$-representability problem for
fermions~\cite{anti}.

Using {\em concurrence}~\cite{ww} as a bipartite entanglement measure,
Coffman-Kundu-Wootters(CKW) provided a mathematical characterization of MoE in three-qubit systems as an inequality~\cite{ckw},
which was generalized for arbitrary multi-qubit systems~\cite{ov}.
As a dual concept of MoE, a {\em polygamy} inequality of multi-qubit entanglement
was established in terms of {\em Concurrence of Assistance}(CoA).
Later, it was shown that the monogamy and polygamy inequalities of multi-qubit entanglement can also be established
by using other entropy-based entanglement measures such as R\'enyi, Tsallis and unified entropies~\cite{KSRenyi, KimT, KSU}.

Recently, a different kind of monogamous relation in multi-qubit entanglement was proposed by using
concurrence and CoA~\cite{ZF15}. Whereas the CKW-type monogamy inequalities of multi-qubit entanglement provide a lower bound
of bipartite entanglement between one qubit subsystem and the rest qubits in terms of two-qubit entanglement, the new kind of
monogamy relations in~\cite{ZF15} provide bounds of bipartite entanglement between a two-qubit subsystem and the rest in multi-qubit systems
in terms of two-qubit concurrence and CoA.

Here, we provide generalized entanglement constraints in multi-qubit systems in terms of
Tsallis entropy for a selective choice of the real parameter $q$.
Using quantum Tsallis entropy of order $q$, namely {\em Tsallis-$q$ entropy},
we first show that the CKW-type monogamy inequality of multi-qubit entanglement
can have a generalized form for $q=2$ or $3$. This generalized monogamy inequality encapsulates
multi-qubit CKW-type monogamy inequality as a special case. We further provide a generalized polygamy inequality
of multi-qubit entanglement in terms of Tsallis-$q$ entropy for $1 \leq q \leq2$ or $3 \leq q \leq 4$, which also
contains multi-qubit polygamy inequality as a special case.

This paper is organized as follows. In Sec.~\ref{Subsec:
definition}, we recall the definition of Tsallis-$q$ entropy, and the bipartite entanglement measure based on
Tsallis entropy, namely Tsallis-$q$ entanglement as well as its dual quantity, Tsallis-$q$ entanglement of assistance(TEoA).
In Sec.~\ref{Subsec: 2formula}, we review the analytic evaluations of Tsallis-$q$ entanglement and TEoA in two-qubit systems based on their functional
relations with concurrence, and we further review the monogamy and polygamy inequalities of multi-qubit entanglement in terms of Tsallis-$q$ entanglement and TEoA
in Sec.~\ref{Sec: monopoly}.
In Sec.~\ref{Sec: gmonoT}, we provide generalized monogamy and polygamy inequalities of multi-qubit entanglement in terms of Tsallis-$q$ entanglement and TEoA,
and we summarize our results in Sec.~\ref{Conclusion}.


\section{Tsallis-$q$ Entanglement}
\label{Sec: Tqentanglement}
\subsection{Definition}
\label{Subsec: definition}

Using a generalized logarithmic function with respect to the parameter $q$,
\begin{eqnarray}
\ln _{q} x &=&  \frac {x^{1-q}-1} {1-q},
\label{qlog}
\end{eqnarray}
quantum Tsallis-$q$ entropy for a quantum state $\rho$ is defined as
\begin{align}
S_{q}\left(\rho\right)=-\T \rho ^{q} \ln_{q} \rho = \frac {1-\T\left(\rho ^q\right)}{q-1}
\label{Tqent}
\end{align}
for $q > 0,~q \ne 1$~\cite{lv}.
Although the quantum Tsallis-$q$ entropy has a singularity at $q=1$,
it converges to von Neumann entropy when $q$ tends to $1$~\cite{S_1},
\begin{equation}
\lim_{q\rightarrow 1}S_{q}\left(\rho\right)=-\T\rho \ln \rho=S\left(\rho\right).
\end{equation}

Based on Tsallis-$q$ entropy, a class of bipartite entanglement measures was introduced;
for a bipartite pure state $\ket{\psi}_{AB}$ and each $q>0$,
its {\em Tsallis-$q$ entanglement}~\cite{KimT} is
\begin{equation}
{\mathcal T}_{q}\left(\ket{\psi}_{A|B} \right)=S_{q}(\rho_A),
\label{TEpure}
\end{equation}
where $\rho_A=\T _{B} \ket{\psi}_{AB}\bra{\psi}$ is the reduced
density matrix of $\ket{\psi}_{AB}$ onto subsystem $A$. For a bipartite mixed state $\rho_{AB}$,
its Tsallis-$q$ entanglement is defined via convex-roof extension,
\begin{equation}
{\mathcal T}_{q}\left(\rho_{A|B} \right)=\min \sum_i p_i {\mathcal T}_{q}(\ket{\psi_i}_{A|B}),
\label{TEmixed}
\end{equation}
where the minimum is taken over all possible pure state
decompositions of $\rho_{AB}=\sum_{i}p_i
\ket{\psi_i}_{AB}\bra{\psi_i}$.

Because Tsallis-$q$ entropy converges to von Neumann entropy when $q$ tends to 1,
we have
\begin{align}
\lim_{q\rightarrow1}{\mathcal T}_{q}\left(\rho_{A|B} \right)=E_{\rm f}\left(\rho_{A|B} \right),
\end{align}
where $E_{\rm f}(\rho_{AB})$ is the EoF~\cite{bdsw} of $\rho_{AB}$ defined as
\begin{equation}
E_{\rm f}(\rho_{A|B})=\min \sum_{i}p_i S(\rho^{i}_{A}),\label{eof}
\end{equation}
with the minimization over all possible pure state
decompositions of $\rho_{AB}$,
\begin{equation}
\rho_{AB}=\sum_{i} p_i |\psi^i\rangle_{AB}\langle\psi^i|,
\label{decomp}
\end{equation}
and $\T_{B}|\psi^i\rangle_{AB}\langle\psi^i|=\rho^{i}_{A}$.
In other words, Tsallis-$q$ entanglement is one-parameter generalization of EoF, and
the singularity of ${\mathcal T}_{q}\left(\rho_{AB}\right)$ at $q=1$ can be replaced by $E_{\rm f}(\rho_{AB})$.

As a dual quantity to Tsallis-$q$ entanglement,
{\em Tsallis-$q$ entanglement of Assistance}(TEoA) was defined as~\cite{KimT}
\begin{equation}
{\mathcal T}^a_{q}\left(\rho_{A|B} \right):=\max \sum_i p_i {\mathcal T}_{q}(\ket{\psi_i}_{A|B}),
\label{TEoA}
\end{equation}
where the maximum is taken over all possible pure state
decompositions of $\rho_{AB}$.
Similarly, we have
\begin{align}
\lim_{q\rightarrow1}{\mathcal T}^a_{q}\left(\rho_{A|B}
\right)=E^a\left(\rho_{A|B} \right),
\label{TsallistoEoA}
\end{align}
where $E^a(\rho_{A|B})$ is the {\em Entanglement of Assistance}~(EoA)
of $\rho_{AB}$ defined as~\cite{cohen}
\begin{equation}
E^a(\rho_{A|B})=\max \sum_{i}p_i S(\rho^{i}_{A}). \label{eoa}
\end{equation}
with the maximization over all possible pure state
decompositions of $\rho_{AB}$.

\subsection{Functional relation with concurrence in two-qubit systems}
\label{Subsec: 2formula}
For any bipartite pure state $\ket \psi_{AB}$, its concurrence is defined as~\cite{ww}
\begin{equation}
\mathcal{C}(\ket \psi_{A|B})=\sqrt{2(1-\T\rho^2_A)}, \label{pure
state concurrence}
\end{equation}
where $\rho_A=\T_B(\ket \psi_{AB}\bra \psi)$. For a mixed state
$\rho_{AB}$, its concurrence and concurrence of assistance(CoA) are defined as
\begin{align}
\mathcal{C}(\rho_{A|B})=\min \sum_k p_k \mathcal{C}({\ket
{\psi_k}}_{A|B}), \label{mixed state concurrence}
\end{align}
and
\begin{align}
\mathcal{C}^a(\rho_{A|B})=\max \sum_k p_k \mathcal{C}({\ket
{\psi_k}}_{A|B}),
\label{CoA}
\end{align}
respectively, where the minimum and maximum are taken over all possible pure state
decompositions, $\rho_{AB}=\sum_kp_k{\ket {\psi_k}}_{AB}\bra
{\psi_k}$.

For two-qubit systems, concurrence and CoA are known to have analytic
formulae~\cite{ww}; for any two-qubit state $\rho_{AB}$,
\begin{align}
\mathcal{C}(\rho_{A|B})=&\max\{0, \lambda_1-\lambda_2-\lambda_3-\lambda_4\},
\label{C_formula}
\end{align}
\begin{align}
\mathcal{C}^a(\rho_{A|B})=&\sum_{i=1}^{4}\lambda_i,
\label{Coa_formula}
\end{align}
where $\lambda_i$'s are the eigenvalues, in decreasing order, of
$\sqrt{\sqrt{\rho_{AB}}\tilde{\rho}_{AB}\sqrt{\rho_{AB}}}$ and
$\tilde{\rho}_{AB}=\sigma_y \otimes\sigma_y
\rho^*_{AB}\sigma_y\otimes\sigma_y$ with the Pauli operator
$\sigma_y$.

Later, it was shown that there is a functional relation between concurrence and Tsallis-$q$ entanglement in two-qubit systems~\cite{KimT}.
For any two-qubit state $\rho_{AB}$ (or bipartite pure state with Schmidt-rank 2), we have
\begin{equation}
{\mathcal T}_{q}\left(\rho_{A|B} \right)=f_{q}\left(\mathcal{C}(\rho_{A|B}) \right),
\label{relationmixed}
\end{equation}
for $1 \leq q \leq4$ where $f_{q}(x)$ is a monotonically increasing convex function defined as
\begin{align}
f_{q}(x)=&\frac{1}{q-1}\left[1-\left(\frac{1+\sqrt{1-x^2}}{2}\right)^{q}
-\left(\frac{1-\sqrt{1-x^2}}{2}\right)^{q}\right]
\label{f_q}
\end{align}
on $0 \leq x \leq 1$~\cite{lim}.

Here we note that the analytic evaluation of concurrence in Eq.~(\ref{C_formula}) together with the functional relations in Eq.~(\ref{relationmixed})
provides us with an analytic formula of Tsallis entanglement in two-qubit systems.
Moreover, the monotonicity and convexity of $f_{q}(x)$ for $1 \leq q \leq 4$ also provide an analytic lower bound of
TEoA,
\begin{equation}
{\mathcal T}^a_{q}\left(\rho_{A|B} \right)\geq f_{q}\left(\mathcal{C}^a(\rho_{A|B}) \right),
\label{relationassis}
\end{equation}
where the equality holds $q=2$ or $3$~\cite{KimT}.

\section{Multi-qubit entanglement constraints in terms of Tsallis entropy}
\label{Sec: monopoly}

The monogamy of a multi-qubit entanglement was shown to have a mathematical
characterization as an inequality; for a multi-qubit state $\rho_{A_1A_2\cdots A_n}$,
\begin{equation}
\mathcal{C}\left(\rho_{A_1|A_2 \cdots A_n}\right)^2  \geq  \mathcal{C}\left(\rho_{A_1|A_2}\right)^2
+\cdots+\mathcal{C}\left(\rho_{A_1|A_n}\right)^2,
\label{nCmono}
\end{equation}
where $\mathcal{C}(\rho_{A_1|A_2\cdots A_n})$ is the
concurrence of $\rho_{A_1A_2\cdots A_n}$ with respect to the
bipartition between $A_1$ and the other qubits, and
$\mathcal{C}(\rho_{A_1|A_i})$ is the concurrence
of the two-qubit reduced density matrix $\rho_{A_1A_i}$ for $i=2,\ldots,
n$~\cite{ckw,ov}. Moreover, the {\em polygamy} (or dual monogamy) inequality of
multi-qubit entanglement was also established using CoA~\cite{gbs} as
\begin{equation}
\left(\mathcal{C}^a\left(\rho_{A_1|A_2 \cdots A_n}\right)\right)^2   \leq  (\mathcal{C}^a\left(\rho_{A_1|A_2}\right))^2
+\cdots+(\mathcal{C}^a\left(\rho_{A_1|A_n}\right))^2,
\label{nCdual}
\end{equation}
where $\mathcal{C}^a(\rho_{A_1|A_2\cdots A_n})$ is the
CoA of $\rho_{A_1A_2\cdots A_n}$ with respect to the
bipartition between $A_1$ and the other qubits, and
$\mathcal{C}^a\left(\rho_{A_1|A_i}\right)$ is the CoA of the two-qubit reduced density
matrix $\rho_{A_1A_i}$ for $i=2,\ldots, n$.

Later, this mathematical characterization of monogamy and polygamy of multi-qubit entanglement
was also proposed in terms of Tsallis entropy, which encapsulate the inequalities
(\ref{nCmono}) and (\ref{nCdual}) as special cases
~\cite{KimT}. Based on the following property of the function
$f_q(x)$ in Eq.~(\ref{f_q}) for $2 \leq q \leq3$,
\begin{equation}
f_q\left(\sqrt{x^2+y^2}\right)\geq f_q(x)+f_q(y), \label{gqmono}
\end{equation}
the Tsallis monogamy inequality of multi-qubit entanglement was proposed as
\begin{equation}
{\mathcal T}_{q}\left( \rho_{A_1|A_2 \cdots A_n}\right)\geq
{\mathcal T}_{q}(\rho_{A_1|A_2}) +\cdots+
{\mathcal T}_{q}(\rho_{A_1|A_n}),
\label{Tmono}
\end{equation}
for $2 \leq q \leq3$.

For the case when $1 \leq q \leq 2$ or $3 \leq q \leq 4$, the function
$f_q(x)$ in Eq.~(\ref{f_q}) also satisfies
\begin{equation}
f_q\left(\sqrt{x^2+y^2}\right)\leq f_q(x)+f_q(y),
\label{gqpoly}
\end{equation}
which leads to the Tsallis polygamy inequality
\begin{equation}
{\mathcal T}^a_{q}\left( \rho_{A_1|A_2 \cdots A_n}\right)\leq
{\mathcal T}^a_{q}(\rho_{A_1|A_2}) +\cdots+{\mathcal
T}^a_{q}(\rho_{A_1|A_n})
\label{Tpoly}
\end{equation}
for any multi-qubit state $\rho_{A_1A_2\cdots A_n}$.

\section{Generalized multi-qubit entanglement constraints in terms of Tsallis entropy}
\label{Sec: gmonoT}
In this section, we provide generalized monogamy and polygamy inequalities of multi-qubit entanglement in terms of Tsallis entanglement and
TEoA. We first recall some properties of Tsallis entropy.

\begin{Prop}(Subadditivity of Tsallis entropy)
For any bipartite quantum state $\rho_{AB}$ with $\rho_A=\T_B \rho_{AB}$, $\rho_B=\T_A \rho_{AB}$, and $q \geq 1 $,
we have
\begin{align}
S_q\left(\rho_{AB}\right)\leq S_q\left(\rho_A\right)+S_q\left(\rho_B\right).
\label{eq: subadd}
\end{align}
\label{subadd}
\end{Prop}

Let us consider a three-party pure state $\ket{\psi}_{ABC}$ and its reduced density matrices $\rho_{BC}=\T_{A}\ket{\psi}_{ABC}\bra{\psi}$,
$\rho_{B}=\T_{AC}\ket{\psi}_{ABC}\bra{\psi}$ and $\rho_{C}=\T_{AB}\ket{\psi}_{ABC}\bra{\psi}$.
For $q \geq 1$, Proposition~\ref{subadd} implies
\begin{align}
S_q\left(\rho_{BC}\right)\leq S_q\left(\rho_B\right)+S_q\left(\rho_C\right).
\label{subadd2}
\end{align}

Because $S_q\left(\rho_{BC}\right)=S_q\left(\rho_A\right)$ and $S_q\left(\rho_{C}\right)=S_q\left(\rho_{AB}\right)$,
Eq.~(\ref{subadd2}) can be rewritten as
\begin{align}
S_q\left(\rho_{A}\right)- S_q\left(\rho_B\right) \leq S_q\left(\rho_{AB}\right),
\label{subadd3}
\end{align}
and similarly, we also have
\begin{align}
S_q\left(\rho_{B}\right)- S_q\left(\rho_A\right) \leq S_q\left(\rho_{AB}\right).
\label{subadd4}
\end{align}
Thus we have the following triangle inequality of Tsallis entropy
\begin{align}
|S_q\left(\rho_{A}\right)- S_q\left(\rho_B\right)| \leq S_q\left(\rho_{AB}\right)\leq S_q\left(\rho_A\right)+S_q\left(\rho_B\right),
\label{triin}
\end{align}
for any bipartite quantum state $\rho_{AB}$ and $q \geq 1$.

\begin{Thm}
For $q=2$ or $3$ and any multi-qubit pure state $\ket{\psi}_{ABC_1C_2\cdots C_n}$, we have
\begin{align}
{\mathcal T}_{q}\left(\ket{\psi}_{AB|C_1C_2\cdots C_n}\right)
\geq \sum_{i=1}^{n}\left[{\mathcal T}_{q}(\rho_{A|C_i})-{\mathcal T}^a_{q}(\rho_{B|C_i})\right],
\label{eq:monothe1}
\end{align}
where $\rho_{AB}=\T_{C_1...C_{n}}(|\psi\rangle\langle\psi|)$,
$\rho_{AC_i}=\T_{BC_1...C_{i-1}C_{i+1}...C_{n}}(|\psi\rangle\langle\psi|)$
and $\rho_{BC_i}=\T_{AC_1...C_{i-1}C_{i+1}...C_{n}}(|\psi\rangle\langle\psi|)$.
\label{monothm1}
\end{Thm}
\begin{proof}
For simplicity, we sometimes denote ${\bf C} =\{C_1, C_2, \cdots, C_n \}$.
From the definition of Tsallis entanglement of $\ket{\psi}_{ABC_1C_2\cdots C_n}$ with respect to the bipartition between
$AB$ and ${\bf C}$, we have
\begin{align}
{\mathcal T}_{q}\left(\ket{\psi}_{AB|{\bf C}}\right)=&S_q\left(\rho_{AB}\right)\nonumber\\
\geq& S_q\left(\rho_{A}\right)- S_q\left(\rho_{B}\right)\nonumber\\
=&{\mathcal T}_{q}\left(\ket{\psi}_{A|B{\bf C}}\right)-{\mathcal T}_{q}\left(\ket{\psi}_{B|A{\bf C}}\right),
\label{triin2}
\end{align}
where the inequality is due to the Inequality (\ref{triin}).

We note that for any pure state $\ket{\psi}_{ABC}$ in a $2\otimes2\otimes d$ quantum system with reduced density matrices
 $\rho_{AB}=\T_{C}\ket{\psi}_{ABC}\bra{\psi}$ and $\rho_{AC}=\T_{B}\ket{\psi}_{ABC}\bra{\psi}$, we have~\cite{22d}
\begin{align}
\mathcal{C}\left(\ket{\psi}_{A|BC}\right)^2=\mathcal{C}^a\left(\rho_{A|B}\right)^2+\mathcal{C}\left(\rho_{A|C}\right)^2.
\label{22d}
\end{align}
For $q=2$ or $3$, Inequalities~(\ref{gqmono}) and (\ref{gqpoly}) imply that
\begin{equation}
f_q\left(\sqrt{x^2+y^2}\right)= f_q(x)+f_q(y),
\label{gqeq}
\end{equation}
therefore
\begin{align}
{\mathcal T}_{q}\left(\ket{\psi}_{A|B{\bf C}}\right)=&f_q\left(\mathcal{C}\left(\ket{\psi}_{A|B{\bf C}}\right)\right)\nonumber\\
=&f_q\left(\sqrt{\mathcal{C}^a\left(\rho_{A|B}\right)^2+\mathcal{C}\left(\rho_{A|{\bf C}}\right)^2}\right)\nonumber\\
=&f_q\left(\mathcal{C}^a\left(\rho_{A|B}\right)\right)+f_q\left(\mathcal{C}\left(\rho_{A|{\bf C}}\right)\right),
\label{ABC1}
\end{align}
where the last equality is due to Eq.~(\ref{gqeq}).
Moreover, we also have
\begin{align}
{\mathcal T}_{q}\left(\ket{\psi}_{B|A{\bf C}}\right)=&f_q\left(\mathcal{C}\left(\ket{\psi}_{B|A{\bf C}}\right)\right)\nonumber\\
\leq&f_q\left(\sqrt{\mathcal{C}^a\left(\rho_{A|B}\right)^2+\sum_{i=1}^{n}\mathcal{C}^a\left(\rho_{B|C_i}\right)^2}\right)\nonumber\\
=&f_q\left(\mathcal{C}^a\left(\rho_{A|B}\right)\right)+f_q\left(\sqrt{\sum_{i=1}^{n}\mathcal{C}^a\left(\rho_{B|C_i}\right)^2}\right),
\label{BAC1}
\end{align}
where the first inequality is due to Inequality~(\ref{nCdual}) and the monotonicity of $f_q(x)$ and the last equality is from Eq.~(\ref{gqeq}).

Eq.~(\ref{ABC1}) and Inequality~(\ref{BAC1}) imply that
\begin{align}
{\mathcal T}_{q}\left(\ket{\psi}_{A|B{\bf C}}\right)&-{\mathcal T}_{q}\left(\ket{\psi}_{B|A{\bf C}}\right)\nonumber\\
&\geq f_q\left(\mathcal{C}\left(\rho_{A|{\bf C}}\right)\right)-f_q\left(\sqrt{\sum_{i=1}^{n}\mathcal{C}^a\left(\rho_{B|C_i}\right)^2}\right).
\label{abcacba}
\end{align}
Here we note that
\begin{align}
f_q\left(\mathcal{C}\left(\rho_{A|{\bf C}}\right)\right)\geq& f_q\left(\sqrt{\sum_{i=1}^{n}\mathcal{C}\left(\rho_{A|C_i}\right)^2}\right)\nonumber\\
=&\sum_{i=1}^{n}f_q\left(\mathcal{C}\left(\rho_{A|C_i}\right)\right)\nonumber\\
=&\sum_{i=1}^{n}\mathcal{T}_{q}\left(\rho_{A|C_i}\right),
\label{AC1}
\end{align}
where the first inequality is due to Inequality~(\ref{nCmono}) and the monotonicity of $f_q(x)$,
the first equality is from the iterative use of Eq.~(\ref{gqeq}), and the last equality is from the functional relation of two-qubit
concurrence and Tsallis entanglement in Eq.~(\ref{relationmixed}).
Moreover, we also have
\begin{align}
f_q\left(\sqrt{\sum_{i=1}^{n}\mathcal{C}^a\left(\rho_{B|C_i}\right)^2}\right)
=&\sum_{i=1}^{n}f_q\left(\mathcal{C}^a\left(\rho_{B|C_i}\right)\right)\nonumber\\
\leq&\sum_{i=1}^{n}\mathcal{T}^a_{q}\left(\rho_{B|C_i}\right),
\label{BCi}
\end{align}
where the first equality is from the iterative use of Eq.~(\ref{gqeq}), and the last inequality is from Inequality~(\ref{relationassis}).

From Inequalities~(\ref{abcacba}), (\ref{AC1}) and (\ref{BCi}), we have
\begin{align}
{\mathcal T}_{q}\left(\ket{\psi}_{A|B{\bf C}}\right)&-{\mathcal T}_{q}\left(\ket{\psi}_{B|A{\bf C}}\right)\nonumber\\
&\geq \sum_{i=1}^{n}\mathcal{T}_{q}\left(\rho_{A|C_i}\right)-\sum_{i=1}^{n}\mathcal{T}^a_{q}\left(\rho_{B|C_i}\right),
\label{abcacba2}
\end{align}
which, together with Inequality~(\ref{triin2}), completes the proof.
\end{proof}

Theorem~\ref{monothm1} provides a monogamy-type lower bound of multi-qubit entanglement between two-qubit subsystem $AB$ and
the other $n$-qubit subsystem $C_1C_2\cdots C_n$
in terms of two-qubit entanglements inherent there.
For the case when one-qubit subsystem $B$ is separable from other qubits, Inequality~(\ref{eq:monothe1})
reduces to the CKW-type monogamy inequality in ~(\ref{Tmono}),
thus Theorem~\ref{monothm1} provides a generalized monogamy relation of multi-qubit entanglement in terms of Tsallis entropy.
The lower bound provided in Theorem~\ref{monothm1} is analytically obtainable due to the analytic evaluation of
two-qubit concurrence and CoA as well as their functional relation with Tsallis entanglement provided in Eq.~(\ref{relationmixed})
and Inequality~(\ref{relationassis}).

Now, we present a generalized polygamy relation of multi-qubit entanglement in terms of TEoA. We first
provide the following theorem, which shows a reciprocal relation of TEoA in three-party quantum systems.
\begin{Thm}
For $q\geq1$ any three-party quantum state $\rho_{ABC}$, we have
\begin{align}
{\mathcal T}_{q}^a\left(\rho_{A|BC}\right)
\leq& {\mathcal T}^a_{q}\left(\rho_{B|AC}\right)+{\mathcal T}^a_{q}\left(\rho_{C|AB}\right).
\label{eq: genpoly1}
\end{align}
\label{thm: genpoly1}
\end{Thm}
\begin{proof}
Let
\begin{align}
\rho_{ABC}=\sum_{j}p_j\ket{\psi_j}_{ABC}\bra{\psi_j}
\label{opt}
\end{align}
be an optimal decomposition realizing ${\mathcal T}_{q}^a\left(\rho_{A|BC}\right)$, that is,
\begin{align}
{\mathcal T}_{q}^a\left(\rho_{A|BC}\right)=\sum_{j}p_j{\mathcal T}_{q}\left(\ket{\psi_j}_{A|BC}\right).
\label{optdecT}
\end{align}

For each pure state $\ket{\psi_j}_{ABC}$ in the decomposition~(\ref{opt})
with $\rho^j_{BC}=\T_{A}\ket{\psi_j}_{ABC}\bra{\psi_j}$, $\rho^j_{B}=\T_{AC}\ket{\psi_j}_{ABC}\bra{\psi_j}$ and
$\rho^j_{C}=\T_{AB}\ket{\psi_j}_{ABC}\bra{\psi_j}$, we have
\begin{align}
{\mathcal T}_{q}\left(\ket{\psi_j}_{A|BC}\right)=&S_q\left(\rho^j_{BC}\right)\nonumber\\
\leq & S_q\left(\rho^j_{B}\right)+S_q\left(\rho^j_{C}\right)\nonumber\\
=&{\mathcal T}_{q}\left(\ket{\psi_j}_{B|AC}\right)+{\mathcal T}_{q}\left(\ket{\psi_j}_{C|AB}\right),
\label{subj}
\end{align}
where the inequality is due to the subadditivity of Tsallis entropy in Proposition~\ref{subadd}.

Now we have
\begin{align}
{\mathcal T}_{q}^a\left(\rho_{A|BC}\right)=&\sum_{j}p_j{\mathcal T}_{q}\left(\ket{\psi_j}_{A|BC}\right)\nonumber\\
\leq & \sum_{j}p_j{\mathcal T}_{q}\left(\ket{\psi_j}_{B|AC}\right)+\sum_{j}p_j{\mathcal T}_{q}\left(\ket{\psi_j}_{C|AB}\right)\nonumber\\
\leq& {\mathcal T}^a_{q}\left(\rho_{B|AC}\right)+{\mathcal T}^a_{q}\left(\rho_{C|AB}\right),
\label{genpoly1}
\end{align}
where the first inequality is from Inequality~(\ref{subj}), and the second inequality is due to the definition of TEoA.
\end{proof}

Theorem~\ref{thm: genpoly1} shows the reciprocal relation of TEoA in three-party quantum systems;
the sum of two TEoA's with respect to two possible bipartition(B|AC and C|AB)  always bounds the TEoA with respect to the remaining bipartition (A|BC).
Moreover, the iterative use of Inequality~(\ref{eq: genpoly1}) naturally leads us to the generalization of Theorem~\ref{thm: genpoly1} into
multi-party quantum systems.

\begin{Cor}
For $q\geq1$ and any multi-party quantum state $\rho_{A_1A_2\cdots A_n}$,
\begin{align}
{\mathcal T}_{q}^a\left(\rho_{A_1|A_2\cdots A_n}\right)
\leq& \sum_{i=2}^{n}{\mathcal T}^a_{q}\left(\rho_{A_i|A_1\cdots\widehat{A_i}\cdots A_n}\right),
\label{eq: genpoly3}
\end{align}
where
\begin{align}
{\mathcal T}^a_{q}\left(\rho_{A_i|A_1\cdots\widehat{A_i}\cdots A_n}\right)={\mathcal T}^a_{q}\left(\rho_{A_i|A_1\cdots A_{i-1}A_{i+1}\cdots A_n}\right)
\end{align}
for each $i=1, \cdots, n$.
\label{racmulti}
\end{Cor}

The following corollary presents a generalized polygamy relation of multi-qubit systems in terms of TEoA.
\begin{Cor}
For $1 \leq q \leq 2$ or  $3 \leq q \leq 4$ and any multi-qubit state $\rho_{ABC_1C_2\cdots C_n}$,
we have
\begin{align}
{\mathcal T}_{q}^a\left(\rho_{AB|C_1C_2\cdots C_n}\right)
\leq &2 {\mathcal T}_{q}^a\left(\rho_{A|B}\right)\nonumber\\
&+\sum_{i=1}^{n}\left[\mathcal{T}^a_{q}\left(\rho_{A|C_i}\right)
+\mathcal{T}^a_{q}\left(\rho_{B|C_i}\right)\right].
\label{eq: genpoly2}
\end{align}
\label{Cor: genpoly2}
\end{Cor}
\begin{proof}
By considering $\rho_{ABC_1C_2\cdots C_n}$ as a three-party quantum state $\rho_{AB{\bf C}}$ with ${\bf C}=C_1C_2\cdots C_n$,
Theorem~\ref{thm: genpoly1} leads us to
\begin{align}
{\mathcal T}_{q}^a\left(\rho_{AB|{\bf C}}\right)
\leq& {\mathcal T}^a_{q}\left(\rho_{A|B{\bf C}}\right)+{\mathcal T}^a_{q}\left(\rho_{B|A{\bf C}}\right).
\label{genpoly2}
\end{align}

Form the multi-qubit Tsallis polygamy inequality in (\ref{Tpoly}), we have
\begin{align}
{\mathcal T}^a_{q}\left(\rho_{A|B{\bf C}}\right)&\leq {\mathcal T}_{q}^a\left(\rho_{A|B}\right)+\sum_{i=1}^{n}\mathcal{T}^a_{q}\left(\rho_{A|C_i}\right)\nonumber\\
{\mathcal T}^a_{q}\left(\rho_{B|A{\bf C}}\right)&\leq {\mathcal T}_{q}^a\left(\rho_{B|A}\right)+\sum_{i=1}^{n}\mathcal{T}^a_{q}\left(\rho_{B|C_i}\right).
\label{genpoly3}
\end{align}
Inequality~(\ref{genpoly2}) together with Inequalities~(\ref{genpoly3}) lead us to Inequality~(\ref{eq: genpoly2}).
\end{proof}

Corollary~\ref{Cor: genpoly2} provides a polygamy-type upper bound of multi-qubit entanglement between two-qubit subsystem $AB$ and the other $n$-qubit
subsystem $C_1C_2\cdots C_n$ in terms of two-qubit TEoA inherent there. For the case when one-qubit subsystem $B$ is
independent from other qubits (that is, $\rho_{AB{\bf C}}=\rho_{A{\bf C}}\otimes \rho_B$),
Inequality~(\ref{eq: genpoly2}) reduces to the Tsallis polygamy inequality in (\ref{Tpoly}). In other words, Corollary~\ref{Cor: genpoly2} shows a generalized polygamy
relation of multi-qubit entanglement in terms of TEoA.


\section{Conclusion}
\label{Conclusion}

We have provided generalized entanglement constraints in multi-qubit systems in terms of Tsallis-$q$ entanglement and TEoA.
We have shown that the CKW-type monogamy inequality of multi-qubit entanglement can have a generalized form
in terms of Tsallis-$q$ entanglement and TEoA for $q=2$ or $3$. This generalized monogamy inequality encapsulates multi-qubit CKW-type
inequality as a special case. We have further shown a generalized polygamy inequality of multi-qubit entanglement in terms of TEoA for
$1 \leq q \leq 2$ or $3 \leq q \leq 4$, which also contains multi-qubit polygamy inequality as a special case.

Whereas entanglement in bipartite quantum systems has been intensively studied with rich understanding, the situation becomes
far more difficult for the case of multi-party quantum systems, and very few are known for its characterization and quantification.
MoE is a fundamental property of multi-party quantum entanglement, which also provides various applications in quantum information theory.
Thus, it is an important and even necessary task to characterize MoE to understand the whole picture of multi-party quantum entanglement.

Although MoE is a typical property of multipartite quantum entanglement,
it is however about the relation of bipartite entanglements among the parties in multipartite systems.
Thus, it is inevitable and crucial to have a proper way of quantifying bipartite
entanglement for a good description of the monogamy nature in multi-party quantum entanglement.

Our result presented here deals with Tsallis-$q$ entropy, a one-parameter class of entropy functions and provide sufficient
conditions on the choice of the parameter $q$ for generalized monogamy and polygamy relations of multi-qubit entanglement.
Noting the importance of the study on multi-party quantum entanglement,
our result provides a useful methodology to understand the monogamy and polygamy nature of multi-party entanglement.

\section*{Acknowledgments}
This research was supported by Basic Science Research Program through the National Research Foundation of Korea(NRF)
funded by the Ministry of Education, Science and Technology(NRF-2014R1A1A2056678).



\begin{thebibliography}{1}
\bibitem{S_1}
For this reason, we sometimes denote
\begin{equation*}
S_{1}\left(\rho\right)=S\left(\rho\right)
\label{TvN}
\end{equation*}
for any quantum state $\rho$.
\bibitem{lim}
Although $f_{q}(x)$ in Eq.~(\ref{f_q}) is not explicitly defined for $q=1$, it is straightforward to check that
\begin{align*}
\lim_{q\rightarrow1}f_{q}(x)=& H\Bigl({1\over 2} + {1\over 2}\sqrt{1-x^2}\Bigr),
\label{limf_q}
\end{align*}
with the binary entropy function $H(t) = -[t\log t + (1-t)\log(1-t)]$.
\end{thebibliography}
\end{document}